\documentclass{sig-alternate}

\usepackage{amsthm,amssymb,amsmath}
\usepackage{graphics}
\usepackage{tikz}
\usepackage[plainpages=false,pdfpagelabels,colorlinks=true,citecolor=blue,hypertexnames=false]{hyperref}
\usepackage{color}

\newtheorem{theorem}{Theorem}
\newtheorem{prop}[theorem]{Proposition}

\newtheorem{lemma}[theorem]{Lemma}

\newtheorem{defi}[theorem]{Definition}
\newtheorem{example}[theorem]{Example}

\def\qed{\quad\rule{1ex}{1ex}}

\let\set\mathbb

\def\pa{\partial}
\def\d{\mathrm{d}}
\def\ord{\operatorname{ord}}
\def\ann{\operatorname{ann}}

\def\fac#1#2#3{{#1}^{[#2]_{#3}}} 
\def\fac#1#2#3{(#1;#2)_{#3}}
\def\leftpart#1#2{{\left.{#1}\right\lceil_{\kern-1pt#2}}}
\def\rightpart#1#2{{\left.{#1}\right\rceil_{#2}}}

\makeatletter
\makeatletter
\def\doi#1{\gdef\@doi{#1}}\def\@doi{}
\global\copyrightetc{
ACM \the\acmcopyr\ ...\$15.00}
\toappear{\the\boilerplate\par{\confname{\the\conf}} \the\confinfo\par \the\copyrightetc.\ifx\@doi\@empty\else\par\@doi.\fi}
\newfont{\mycrnotice}{ptmr8t at 7pt}
\newfont{\myconfname}{ptmri8t at 7pt}
\let\crnotice\mycrnotice%
\let\confname\myconfname%
\makeatother

\begin{document}

\permission{Permission to make digital or hard copies of all or part of this
work for personal or classroom use is granted without fee provided that copies
are not made or distributed for profit or commercial advantage and that copies
bear this notice and the full citation on the first page. Copyrights for
components of this work owned by others than the author(s) must be
honored. Abstracting with credit is permitted. To copy otherwise, or
republish, to post on servers or to redistribute to lists, requires prior
specific permission and/or a fee. Request permissions from
Permissions@acm.org.}

\conferenceinfo{ISSAC '14,}{July 23 - 25, 2014, Kobe, Japan\\ 
Copyright is held by the owner/author(s). Publication rights licensed to ACM.}
\crdata{978-1-4503-2501-1/14/07}
\doi{http://dx.doi.org/10.1145/2608628.2608641}


\title{A Generalized Apagodu-Zeilberger~Algorithm}

\numberofauthors{3}

\author{%
 \alignauthor
 \leavevmode
 \mathstrut Shaoshi Chen\titlenote{Supported by a
NSFC grant (91118001) and a 973 project
(2011CB302401).}\\[\smallskipamount]
  \affaddr{\leavevmode\mathstrut  KLMM, AMSS / Chinese Academy of Sciences}\\
  \affaddr{\leavevmode\mathstrut 100190 Beijing, China}\\
  \affaddr{\leavevmode\mathstrut schen@amss.ac.cn}
 \and
 \mathstrut Manuel Kauers\titlenote{Supported by the Austrian Science Fund (FWF) grant Y464-N18.}\\[\smallskipamount]
  \affaddr{\mathstrut RISC / Joh. Kepler University}\\
  \affaddr{\mathstrut 4040 Linz, Austria}\\
  \affaddr{\mathstrut mkauers@risc.jku.at}
 \and
 \mathstrut Christoph Koutschan\\[\smallskipamount]
  \affaddr{\mathstrut RICAM / Austrian Academy of Sciences}\\
  \affaddr{\mathstrut 4040 Linz, Austria}\\
  \affaddr{\mathstrut christoph.koutschan@ricam.oeaw.ac.at}
}

\maketitle
\begin{abstract}
  The Apagodu-Zeilberger algorithm can be used for computing annihilating
  operators for definite sums over hypergeometric terms, or for definite
  integrals over hyperexponential functions. In this paper, we propose a
  generalization of this algorithm which is applicable to arbitrary $\partial$-finite functions.
  In analogy to the hypergeometric case, we
  introduce the notion of proper $\partial$-finite functions. We show that the
  algorithm always succeeds for these functions, and we give a tight a priori
  bound for the order of the output operator.
\end{abstract}


\category{I.1.2}{Computing Methodologies}{Symbolic and Algebraic Manipulation}[Algorithms]


\terms{Algorithms}


\keywords{Symbolic summation, symbolic integration, $\partial$-finite function,
holonomic function, Ore algebra, creative telescoping}


\section{Introduction}

We consider the problem of creative telescoping: given a function $f(x,y)$, the task consists
in finding a linear operator $T\neq0$ in $x$ only, called a \emph{telescoper}, and another
operator $C$ possibly involving both $x$ and~$y$, called a \emph{certificate} for~$T$, such
that $T-\partial_y C$ annihilates the given function~$f(x,y)$. Here $\partial_y$ may be
for example the partial derivation~$\frac{\d}{\d y}$ or the forward difference operator
$\Delta_y$ with respect to~$y$.

Pairs $(T,C)$ are used for solving summation and integration problems. For
example, given a definite integral $F(x)=\int_0^1 f(x,y)\,\d y$ depending on a free
parameter~$x$, we may want to compute a creative telescoping relation
\[
 \Bigl(t_0+t_1\pa_x+\cdots+t_r\pa_x^r\Bigr)\cdot f(x,y) = \pa_y C \cdot f(x,y),
\]
where~$\pa_x, \pa_y$ are the partial derivations~$\frac{\d}{\d x},
\frac{\d}{\d y}$, respectively. By integrating both sides of the relation
above w.r.t.~$y$, we obtain an inhomogeneous linear differential equation
\[
  t_0 F(x) + t_1 F'(x) + \cdots + t_r F^{(r)}(x) = \bigl[\,C \cdot f(x,y)\,\bigr]_{y=0}^1
\]
for the integral. This equation can then be processed further by other algorithms,
for example to find closed form representations or asymptotic expansions for~$F(x)$.

Algorithms for computing creative telescoping pairs $(T,C)$ are known for
various classes of functions~$f(x,y)$. For hypergeometric terms, which
satisfy two first-order recurrence equations in $x$ and~$y$ respectively, the problem is
solved by Zeilberger's algorithm~\cite{Zeilberger1990c, Zeilberger1991}. An analogous algorithm for
hyperexponential functions, which satisfy two first-order differential
equations in $x$ and~$y$ respectively, was given by Almkvist and
Zeilberger~\cite{Almkvist1990}. In 1997, Chyzak~\cite{Chyzak2000} generalized these algorithms to
the case of general holonomic $\partial$-finite functions~$f(x,y)$, which are solutions of
systems of higher-order recurrence and/or differential equations, see Section~\ref{def:dfinite}
for a definition. For a detailed introduction to creative telescoping in the context
of holonomic functions, see~\cite{Koutschan13a,Chyzak14}. 

In 2005, Apagodu (formerly ``Mohamud Mohammed'') and Zeilberger~\cite{AZ2005} proposed an interesting variation of
Zeil\-ber\-ger's original algorithm for hypergeometric terms. This algorithm,
sketched in Section~\ref{sec:az} below, is easier to implement than Zeilberger's
original algorithm, it requires less computation time, and it gives rise to
good bounds for the order of the telescopers. A similar approach to compute
telescopers for general holonomic $\partial$-finite functions was proposed
and implemented in~\cite{Koutschan2010}; it proved superior to Chyzak's
algorithm in many examples from applications, but used some heuristics and thus lacked rigor.
In particular, no bounds concerning the telescoper were given there.

In the present paper, we want to do with the Apagodu-Zeilberger algorithm what
Chyzak did with the original Zeilberger algorithm: we extend it to a more
general setting. The setting is more general in two senses. First, we drop the
condition that the input is specified by first-order equations and instead cover
arbitrary $\partial$-finite input. Second, we do not restrict to the shift
and/or differential case but formulate the result in the language of Ore
algebras.  In this general context, we lose the property known for the
differential case that a creative telescoping pair $(T,C)$ always
exists. Therefore, in analogy to the hypergeometric case, we introduce the
notion of \emph{proper $\partial$-finite} functions, and give an explicit upper
bound on the order of telescopers for such functions. Good bounds are
useful in practice as they allow to compute telescoper and certificate 
without having to loop over the order of the ansatz operator (as it is
done, for example, in Zeilberger's algorithm).

\section{Preliminaries}

\subsection{Ore Algebras}\label{sec:ore}

The operator algebras we are going to work with were introduced
by Ore in 1933~\cite{Ore33}. They provide a common framework for representing
linear differential equations and linear ($q$-) difference equations;
the coefficients of these equations may be polynomials or rational functions,
for example.

Let $K$ be a field with $\set Q\subseteq K$.  Let $\sigma_x,\sigma_y\colon K(x,y)\to K(x,y)$
be field automorphisms with
$\sigma_x\sigma_y=\sigma_y\sigma_x$, and let $\delta_x,\delta_y\colon K(x,y)\to
K(x,y)$ be $K$-linear maps satisfying
$\delta_x(ab)=\delta_x(a)b+\sigma_x(a)\delta_x(b)$ and
$\delta_y(ab)=\delta_y(a)b+\sigma_y(a)\delta_y(b)$ for all $a,b\in K(x,y)$.  The
set $\set A=K(x,y)[\partial_x,\partial_y]$ of all bivariate polynomials in
$\partial_x,\partial_y$ with the usual addition, and with the unique
noncommutative multiplication satisfying
$\partial_x\partial_y=\partial_y\partial_x$ and $\partial_x a =
\sigma_x(a)\partial_x + \delta_x(a)$ and $\partial_y a = \sigma_y(a)\partial_y +
\delta_y(a)$ for all $a\in K(x,y)$ is an \emph{Ore algebra}~\cite{ChyzakSalvy1998}.
All Ore algebras appearing in this paper will be of this form.

Note that $\delta_x(\frac 1a)=-\frac{\delta_x(a)}{a\sigma_x(a)}$ for all $a\in
K(x,y)\setminus\{0\}$, and likewise for~$\delta_y$.

We assume that $\sigma_x,\sigma_y,\delta_x,\delta_y$ map polynomials 
to polynomials. Moreover, we assume that
$\deg_x(\sigma_x(p))=\deg_x(p)$, $\deg_y(\sigma_x(p))=\deg_y(p)$,
$\deg_x(\delta_x(p))\leq\deg_x(p)-1$ and that $\deg_y(\delta_x(p))\leq\deg_y(p)$ for
all $p\in K[x,y]$; likewise for $\sigma_y,\delta_y$.

\subsection{{\large $\partial$}-Finite Functions}\label{def:dfinite}
Many special functions used in mathematics and physics are solutions of systems of
linear differential and/or recurrence equations. Hypergeometric terms are functions that satisfy
a system of first-order linear recurrence equations and their continuous analogue
are hyperexponential functions. Their generalization to functions that satisfy a system
of higher-order equations leads to the concept of $\partial$-finite functions.

We let the Ore algebra~$\set A$ act on an appropriate space $F$ of ``functions''
by defining an operation $\cdot\colon\set A\times F\to F$; in particular,
one has to fix the result of applying $\partial_x$ and $\partial_y$ to a function.
The operation of applying an Ore operator $P\in\set A$ to a function~$f\in F$
turns $F$ into a left $\set A$-module. We define the
\emph{annihilator} (w.r.t.~$\set A$) of a function~$f$ as the set
$\{P\in\set A\mid P\cdot f=0\}$, denoted $\ann_{\set A}(f)$;
it is easy to verify that it is a left ideal in~$\set A$. For every left ideal
$\mathfrak{a}\subseteq\set A$ the quotient algebra $\set A/\mathfrak{a}$ is a
$K(x,y)$-vector space.

A left ideal $\mathfrak{a}\subseteq\set A$ is called \emph{zero-dimensional} or
\emph{$\partial$-finite} if $\dim_{K(x,y)}(\set A/\mathfrak{a})$ is finite.
A function~$f$ is called \emph{$\partial$-finite} (w.r.t.~$\set A$) if
$\ann_{\set A}(f)$ is a zero-dimensional left ideal.

\subsection{Left and Right Borders}\label{sec:borders}

Part of the additional generality provided in this paper comes at the expense of
a somewhat involved notation, which we now introduce.

For $a\in K(x,y)$ and $i\in\set N$, write $\fac aiy :=
\prod_{j=0}^{i-1}\sigma_y^j(a)$. Let~$p$ be a polynomial in
$K(x)[y]\setminus\{0\}$.  Choose~$n$ to be the largest nonnegative integer such
that there is a monic factor~$p_n$ in $K(x)[y]\setminus K(x)$ with
$\deg_y(p_n)$ as large as possible and $\fac {p_n}ny$ dividing~$p$.  We repeat
this process for $p/\fac {p_n}ny$ until obtaining a constant~$c\in K(x)$.  In
this way, the polynomial~$p$ can be written uniquely as
$p=c\prod_{i=1}^n\fac{p_i}{i}y$ with $c\in K(x)\setminus\{0\}$ and
$p_1,\dots,p_n\in K(x)[y]$ monic such that $\deg_y(p_n)>0$. When
$\sigma_y=\mathrm{id}$ this is the squarefree decomposition of~$p$ in~$y$, and
when $\sigma_y(y)=y+1$ it is the greatest factorial
factorization~\cite{Paule1995b} in~$y$, where the falling factorials in the
original definition are expressed in terms of rising factorials.  Define
$\leftpart py:=\prod_{i=1}^n p_i$ (left border of~$p$) and $\rightpart
py:=\prod_{i=1}^n \sigma_y^{i-1}(p_i)$ (right border of~$p$). When
$\sigma_y=\mathrm{id}$, both the left border $\leftpart py$ and the
right border $\rightpart py$ are equal to the squarefree part
of~$p$.  By definition, we have $p\,\sigma_y(\rightpart
py)=c\prod_{i=1}^n\fac{p_i}{i+1}y$, $\frac p{\rightpart
  py}=c\prod_{i=1}^n\fac{p_i}{i-1}y$, and the equality
\begin{equation}\label{eq:lrborders}
  p \, \sigma_y(\rightpart py) = {\leftpart py} \, \sigma_y(p).
\end{equation}
The notations $\fac aix$ (for $a\in K(x,y)$) and $\leftpart px,\rightpart px$ (for $p\in
K(y)[x]\setminus\{0\}$) are defined analogously.

\begin{lemma}\label{lemma:delta}
For $a\in K(x,y)\setminus\{0\}$ we have
\begin{equation}\label{eq:deltafac}
  \delta_y\big(\fac any\big) = \frac{\fac any}{a} \sum_{i=0}^{n-1}\delta_y\big(\sigma_y^i(a)\big).
\end{equation}
Moreover, for $p\in K(x)[y]\setminus\{0\}$ we have that $p \mathrel{\big|} \leftpart py \, \delta_y(p)$,
as polynomials in~$y$. Analogous statements hold when switching the roles of $x$ and~$y$.
\end{lemma}
\begin{proof}
The proof uses the general product rule for $\delta_y$,
\begin{equation}\label{eq:prodrule}
  \delta_y\bigg(\prod_{i=1}^n a_i\bigg) =
    \sum_{i=1}^n \bigg(\prod_{j=1}^{i-1} \sigma_y(a_j)\bigg) \delta_y(a_i)
      \bigg(\prod_{j=i+1}^n a_j\bigg),
\end{equation}
where $a_1,\dots,a_n$ are arbitrary elements in $K(x,y)$;
the formula can be verified by an easy induction argument.
Applying~\eqref{eq:prodrule} to $\fac any$ proves the first assertion:
\begin{align*}
  \delta_y\big(\fac any\big) &= \delta_y\bigg(\prod_{i=1}^n \sigma_y^{i-1}(a)\bigg) \\
  &= \sum_{i=1}^n \bigg(\prod_{j=1}^{i-1} \sigma_y^j(a)\bigg) \delta_y\big(\sigma_y^{i-1}(a)\big)
    \bigg(\prod_{j=i+1}^n \sigma_y^{j-1}(a)\bigg) \\
  &= \fac{\sigma_y(a)}{n-1}y \sum_{i=0}^{n-1}\delta_y\big(\sigma_y^i(a)\big)
\end{align*}
To prove the second assertion, assume that $p$ is given in the unique form
$c\prod_{i=1}^n\fac{p_i}{i}y$ as above. Once again, the product
rule~\eqref{eq:prodrule} is employed, in combination with the first assertion:
\begin{align*}
  \delta_y(p) &= \delta_y\bigg(c\prod_{i=1}^n \fac{p_i}iy\bigg) \\
  &= c\sum_{i=1}^n \bigg(\prod_{j=1}^{i-1} \sigma_y\big(\fac{p_j}jy\big)\bigg)
     \delta_y\big(\fac{p_i}iy\big) \bigg(\prod_{j=i+1}^n \fac{p_j}jy \bigg) \\
  &= \frac{c\,p}{\leftpart py} \sum_{i=1}^n \bigg(\prod_{j=1}^{i-1}\sigma_y^j(p_j)\bigg)
       \bigg(\sum_{k=0}^{i-1}\delta_y\big(\sigma_y^k(p_i)\big)\bigg)
       \bigg(\prod_{j=i+1}^n\! p_j\bigg).
\end{align*}
Note that $\frac{p}{\leftpart py}=\prod_{j=2}^n \sigma_y\big(\fac{p_j}{j-1}y\big)$
in the last line. This completes the proof.
\end{proof}

\subsection{The Apagodu-Zeilberger Algorithm}\label{sec:az}

The Apagodu-Zeilberger algorithm~\cite{AZ2005} solves the same problem as
Zeilberger's algorithm~\cite{Zeilberger1991}: creative telescoping for proper hypergeometric terms.
Before generalizing this algorithm to general $\partial$-finite functions, let us summarize the
reasoning behind it at a simple example. Consider
the hypergeometric term $h(x,y):=\frac1{\Gamma(ax+by)}$ for two positive
integers $a,b\in\set N$. We want to find
$T=t_0+t_1\partial_x+\cdots+t_r\partial_x^r\in K(x)[\partial_x]\setminus\{0\}$
(a~\emph{telescoper}) and $C\in K(x,y)$ (a~\emph{certificate}) such that
\[
  T\cdot h(x,y) = \partial_y C \cdot h(x,y),
\]
where $\partial_x$ denotes the shift operator with respect to~$x$ (i.e. $\sigma_x(x)=x+1$, $\delta_x=0$)
and $\partial_y$ denotes the forward difference with respect to~$y$ (i.e. $\sigma_y(y)=y+1$, $\delta_y(y)=1$).

By $\Gamma(ax+by+ia)=(ax+by)(ax+by+1)\cdots(ax+by+ia-1)\cdot \Gamma(ax+by)$ for all $i\geq 0$, we have
\[
  T\cdot h(x,y) = \frac{u}{(ax+by)\cdots(ax+by+ra-1)} h(x,y)
\]
for some polynomial $u$ of $y$-degree $ra$ whose coefficients are linear
combinations of the undetermined coefficients~$t_i$. For the choice
\[
  C=\frac{c_0+c_1y+\cdots+c_sy^s}{(ax+by)(ax+by+1)\cdots(ax+by+ra-b-1)},
\]
we obtain
\[
  \partial_y C\cdot h(x, y) = \frac{v}{(ax+by)\cdots(ax+by+ra-1)} h(x,y)
\]
for some polynomial $v$ of $y$-degree $s + b$. The denominators on both sides agree, and if we take
$s=ra-b$, so do the degrees (provided that $ra\geq s$).
Coefficient comparison yields a linear homogeneous system with
$ra+1$ equations and $(r+1)+(ra-b+1)$ variables (the $t_i$'s and the $c_j$'s). As soon as $r\geq b$,
this system has a nontrivial solution.

A telescoper $T$ coming from such a nontrivial solution cannot be zero, for if
it were, then also $\partial_y C\cdot h(x,y)$ would be zero, and then $C\cdot
h(x,y)$ would be constant with respect to~$y$, which is not the case because
$C$ is a nonzero rational function but $h(x,y)$ is not rational.

Similar calculations can be carried out for when $a$ or $b$ are negative. By
plugging all of them together, Apagodu and Zeilberger~\cite{AZ2005} show that a
(non-rational) proper hypergeometric term
\[
  p\,\alpha^x\beta^y\prod_{m=1}^M
    \frac{\Gamma(a_mx+a_m'y+a_m'')\Gamma(b_mx-b_m'y+b_m'')}
         {\Gamma(u_mx+u_m'y+u_m'')\Gamma(v_mx-v_m'y+v_m'')}
\]
($p\in K[x,y]$, $M\in\set N$, $\alpha,\beta,a_m'',b_m'',u_m'',v_m''\in K$,
$a_m$, $a_m'$, $b_m$, $b_m'$, $u_m$, $u_m'$, $v_m$, $v_m'\in\set N$) admits a nonzero
telescoper of order at most $\max\bigl\{\sum_{m=1}^M (a_m'+v_m'),\sum_{m=1}^M
(u_m'+b_m')\bigr\}$. A refinement of this bound, including the $x$-degree
of the telescoper, is given in~\cite{ChenKauers12b}.

In the differential case, they find~\cite{Apagodu2006} that a (non-rational) hyperexponential function
\[
  p \, \exp\!\left(\frac ab\right)\prod_{m=1}^M q_m^{e_m}
\]
($a,b,p,q_1,\dots,q_M\in K(x)[y]$, $e_1,\dots,e_m\in K$) admits a telescoper of
order at most
$\deg_y(b)+\max\{\deg_y(a),\deg_y(b)\}+\sum_{m=1}^M\deg_y(q_m)-1$.
In~\cite[Thm. 14]{ChenKauers12a} it is shown that this bound can be improved
by replacing the first term $\deg_y(b)$ by the $y$-degree of the
squarefree part of~$b$, and that when the term is a rational function the bound
increases by~$1$.
A further improvement is given in \cite[Sec. 6.2]{BostanChenChyzakLiXin13}.

\section{The General Case}\label{sec:general}

Let $\set A=K(x,y)[\partial_x,\partial_y]$ be an Ore algebra as introduced in
Section~\ref{sec:ore} and $\mathfrak{a}\subseteq\set A$ be a $\partial$-finite
ideal.  Further let $B=\{b_1,\dots,b_n\}$ be a $K(x,y)$-basis of $\set
A/\mathfrak{a}$, so that every element of $\set A/\mathfrak{a}$ can be written
uniquely in the form $wb=\sum_{i=1}^nw_ib_i$ for some vector
$w=(w_1,\dots,w_n)\in K(x,y)^n$ and $b=(b_1,\dots,b_n)^T$. We say that the
vector $w=(w_1,\dots,w_n)$ represents the element $wb\in\set A/\mathfrak a$.
For all~$b_i\!\in\! B$ we can write
\begin{equation}\label{eq:multmat}
  \partial_x b_i = \sum_{j=1}^n m_{i,j} b_j \quad \text{with~$m_{i,j} \in K(x, y)$}.
\end{equation}
With $M=(m_{i,j})_{1\leq i,j\leq n}\in K(x,y)^{n\times n}$
equation~\eqref{eq:multmat} can be expressed succinctly as $\partial_xb=Mb$,
where the operator $\partial_x$ is applied componentwise.  Applying
$\partial_x$ to an arbitrary element in $\set A/\mathfrak a$ gives
\[
  \partial_x(wb) = \big(\sigma_x(w)\partial_x+\delta_x(w)\big)b =
  \big(\sigma_x(w)M+\delta_x(w)\big)b
\]
where $\sigma_x$ and $\delta_x$ act on the components of~$w\in K(x,y)^{n}$.
As a shorthand notation, we will write the above identity as
$\partial_xw=\sigma_x(w)M+\delta_x(w)$, by defining an action of operators
from~$\set A$ on vectors in $K(x,y)^n$.  Similarly, there is a matrix~$N\in
K(x,y)^{n\times n}$ such that $\partial_yb=Nb$ and $\partial_y w=
\sigma_y(w)N+\delta_y(w)$.  Without loss of generality, we assume that the
basis~$B$ is chosen in such a way that the element $1\in\set A/\mathfrak{a}$
is represented by a \emph{polynomial} vector $e\in K(x)[y]^n$; we call such
bases ($y$-)\emph{admissible}.

The matrices $M$ and $N$ correspond to the rational functions $\partial_x h/h$
and $\partial_y h/h$ in the hyperexponential case, and similarly in the
hypergeometric case. In general, their entries are rational functions. We will
write $M=\frac1uU$, $N=\frac1vV$ where $u,v\in K[x,y]$ and $U,V\in
K[x,y]^{n\times n}$.

\subsection{Telescoper Part}\label{sec:T}

For $r\geq1$, make an ansatz $T=\sum_{i=0}^r t_i\partial_x^i$ for the
telescoper, in which $t_0,\dots,t_r$ stand for undetermined elements
of~$K(x)$, so that $T$ is an element of $K(x)[\partial_x]\subseteq\set A$. We
need to discuss the shape of the vector $Te\in K(x,y)^{n}$, i.e., its
denominator and its numerator degree in~$y$.

\begin{lemma}\label{lemma:T}
Let $e\in K(x)[y]^n$ be some polynomial vector. 
For every $i\geq0$ we have $\partial_x^i e = \frac1{\fac uix}w$
for some vector $w\in K(x)[y]^n$ with
\[
  \deg_y(w)\leq \deg_y(e)+i\max\{\deg_y(u),\deg_y(U)\},
\]
where the degree of a matrix or vector refers to the maximum degree of its components.
\end{lemma}
\begin{proof}
  The claim is evident for $i=0$. Assume it holds for~$i$. Then
  \begin{alignat*}1
    \partial_x^{i+1}e
      &=\partial_x\Bigl(\frac1{\fac uix}w\Bigr)\\
      &=\sigma_x\Bigl(\frac1{\fac uix}w\Bigr)\frac1uU + \delta_x\Bigl(\frac1{\fac uix}w\Bigr)\\
      &=\frac{\sigma_x(w)U}{u\,\sigma_x\big(\fac uix\big)} + \frac{\delta_x(w)}{\fac uix} +
        \delta_x\Bigl(\frac1{\fac uix}\Bigr)\sigma_x(w)\\
      &=\frac{\sigma_x(w)U}{\fac u{i+1}x} + \frac{\delta_x(w)}{\fac uix} 
        - \frac{\delta_x\big(\fac uix\big)\,\sigma_x(w)}{\fac uix\,\sigma_x\big(\fac uix\big)}.
  \end{alignat*}
  The last term in the above line can be simplified by
  \[
    \frac{\delta_x\big(\fac uix\big)}{\fac uix\,\sigma_x\big(\fac uix\big)} =
    \frac1{\fac u{i+1}x}\frac{u\,\delta_x\big(\fac uix\big)}{\fac uix} =
    \frac{\tilde{u}}{\fac u{i+1}x}
  \]
  where Lemma~\ref{lemma:delta} ensures that $\tilde{u}$ is a polynomial in $K[x,y]$.
  Since we assume throughout that $\sigma_x$ and $\delta_x$ do not increase the $y$-degree
  of polynomials, we conclude from~\eqref{eq:deltafac} that $\deg_y(\tilde u)\leq\deg_y(u)$.
  Therefore, we obtain
  \[
    \partial_x^{i+1}e = \frac{\sigma_x(w)U + \sigma_x^i(u)\delta_x(w) - \tilde u \sigma_x(w)}{\fac u{i+1}x},
  \]
  and the whole numerator is bounded in $y$-degree by
  \begin{alignat*}1
    &\max\bigl\{\deg_y(\sigma_x(w))+\deg_y(U), \\
    &\qquad\deg_y(\sigma_x^i(u))+\deg_y(\delta_x(w)),\\
    &\qquad\deg_y(\tilde u)+\deg_y(\sigma_x(w))\bigr\}\\
    &\leq \max\bigl\{\deg_y(w)+\deg_y(U), \deg_y(u)+\deg_y(w), \\
    &\qquad\deg_y(u)+\deg_y(w)\bigr\}\\
    &\leq \max\{\deg_y(U),\deg_y(u)\}+\deg_y(w)\\
    &\leq \deg_y(e)+(i+1)\max\{\deg_y(u),\deg_y(U)\},
  \end{alignat*}
  as claimed.
\end{proof}

By writing $\frac1{\fac uix}=\frac{\fac{\sigma^i(u)}{r-i}x}{\fac urx}$ in the above lemma,
we find that we can write
\[
  Te = \frac1{\fac urx}w
\]
for some vector $w$ whose entries are linear combinations of
$t_0,\dots,t_r$ with coefficients in $K(x)[y]$ bounded in degree by
$\deg_y(e)+r\max\{\deg_y(u),\deg_y(U)\}$.

\subsection{Certificate Part}\label{sec:C}

We need to characterize those certificates $C\in\set A$ for which the vector
$\partial_yCe$ matches a prescribed numerator degree and a prescribed denominator $d\in K(x)[y]$;
as before, let $e$ denote the polynomial
vector representing the element $1\in\set A/\mathfrak{a}$ with respect to the
basis~$B$. It will be convenient to focus on possible numerators and
denominators of the vector $c:=Ce\in K(x,y)^n$.

Let $d\in K(x)[y]$ be the target denominator. It will turn out that factors of
$d$ which also appear in $v$ (the denominator in the $\partial_y$-multiplication matrix)
behave slightly different than other factors. Let us therefore write
\begin{alignat*}1
  d &= \fac{f_1}{p_1}y\cdots\fac{f_m}{p_m}y\,g\\
  v &= \fac{f_1}{q_1}y\cdots\fac{f_m}{q_m}y\,\sigma_y(h)
\end{alignat*}
so that $f_1,\dots,f_m\in K(x)[y]$ are common factors of $d$ and $v$.
Note that we don't impose any coprimeness conditions on the $f_i$'s
with~$g$ and~$h$. Therefore, without loss of generality, we may always
assume that $p_i\geq q_i$, by moving possible overhanging factors of
some $\fac{f_i}{q_i}y$ into~$\sigma_y(h)$.

\begin{lemma}\label{lemma:C}
  Assume that $p_i\geq q_i\geq1$ for $i=1,\dots,m$ and let
  \[
    z = \sigma_y^{-1}\Bigl(\frac{\fac{f_1}{p_1}y\cdots\fac{f_m}{p_m}y}{\fac{f_1}{q_1}y\cdots\fac{f_m}{q_m}y}\Bigr)
     \frac{g}{\rightpart gy}.
  \]
  Note that $z\in K(x)[y]$.
  Let $w\in K(x)[y]^{n}$ be any polynomial vector and consider
  $c=\frac{h}{z}w$.
  Then $\partial_y c=\frac1d\tilde{w}$ for some vector $\tilde{w}\in K(x)[y]^{n}$
  with $\deg_y(\tilde{w})\leq \deg_y(w)+\deg_y(\leftpart gy)+\max\{\deg_y(v)-1,\deg_y(V)\}$.
\end{lemma}
\begin{proof}
  We show that $d\,\partial_yc$ is a polynomial vector with the claimed degree. From
  \[
    \partial_yc=\partial_y \frac{h}{z}w
    =\frac{\sigma_y(h)}{\sigma_y(z)}\sigma_y(w)\frac1vV+\delta_y\Bigl(\frac hzw\Bigr)
  \]
  we get
  \begin{alignat}1\label{eq:ddyC}
    d\,\partial_yc=\frac{d\,\sigma_y(h)}{v\,\sigma_y(z)}\sigma_y(w)V+
    d\,\delta_y\Bigl(\frac hz\Bigr)w + \frac {d\,\sigma_y(h)}{\sigma_y(z)}\delta_y(w).
  \end{alignat}
  A straightforward calculation using~\eqref{eq:lrborders} gives the equality
  $d\sigma_y(h)= \sigma_y(z)v\leftpart gy$ which we employ in the following to replace~$d$.
  The first term in expression~\eqref{eq:ddyC}
  simplifies to $\leftpart gy\sigma_y(w)V$, the $y$-degree of which is bounded by
  $\deg_y(\leftpart gy)+\deg_y(w)+\deg_y(V)$, as claimed. Similarly, the third term
  simplifies to $v\leftpart gy\delta_y(w)$, the $y$-degree of which is bounded
  by $\deg_y(v)+\deg_y(\leftpart gy)+\deg_y(w)-1$, also as claimed. Finally,
  the second term of~\eqref{eq:ddyC} is considered:
  \begin{alignat*}1
    d\,\delta_y\Bigl(\frac hz\Bigr)w
    &=\frac{\sigma_y(z)\,v\leftpart gy}{\sigma_y(h)} \, \frac{\delta_y(h)\,z-h\,\delta_y(z)}{z\,\sigma_y(z)} \, w\\
    &=\frac{v}{\sigma_y(h)}\leftpart gy \delta_y(h)w - \frac{v\,\delta_y(z)\leftpart gy}{z\,\sigma_y(h)}hw.
  \end{alignat*}
  The first term in this expression is a polynomial because $\sigma_y(h)\mid v$.
  Its degree is bounded by $\left(\deg_y(v)-\deg_y(h)\right)+\left(\deg_y(h)-1\right)+
  \deg_y(\leftpart gy)+\deg_y(w) =\deg_y(w)+\deg_y(\leftpart gy)+\deg_y(v)-1$, as
  claimed. Also for the second term the degree count matches the claim. To see
  finally that also this second term is a polynomial in~$y$, write
  \begin{alignat*}1
    z & =
      \bigg(\prod_{i=1}^m\fac{\sigma_y^{q_i-1}(f_i)}{p_i-q_i}y\bigg) g,
  \end{alignat*}
  which implies that~$\leftpart zy$ divides~$\fac{f_1}{q_1}y\cdots\fac{f_m}{q_m}y \leftpart gy$.
  We now write the second term as
  \[
   \frac{v\,\delta_y(z)\leftpart gy}{z\,\sigma_y(h)}hw =
   \frac{\fac{f_1}{q_1}y\cdots\fac{f_m}{q_m}y \leftpart gy}{\leftpart zy} \, \frac{\leftpart zy \, \delta_y(z)}{z}\, hw
 \]
 and observe that $z$ divides $\leftpart zy \, \delta_y(z)$ by Lemma~\ref{lemma:delta}, which concludes the proof.
\end{proof}

\subsection{Proper {\large $\partial$}-finite Ideals}

In order to obtain a bound for the order of the telescoper, we apply
Lemmas~\ref{lemma:T} and~\ref{lemma:C} in such a way that the vectors
$Te$ and $\partial_y c$ match. In particular, we need to match the denominator
and the degree of the numerator. From Section~\ref{sec:T} we know that the
denominator coming from the telescoper part is $\fac urx$, and the $y$-degree of
the numerator is at most $\deg_y(e)+r\max\{\deg_y(u),\deg_y(U)\}$. From
Section~\ref{sec:C} we know how to choose $c$ in such a way that $\partial_yc$
has a prescribed denominator and a given numerator degree. Coefficient
comparison with respect to~$y$
will give a system of linear equations, and we will be able to choose
$r$ in such a way that this system has a solution.

This is the basic idea, but there is a complication. The denominator coming from
the telescoper part is expressed with respect to~$\sigma_x$ while
Lemma~\ref{lemma:C} requires the prescribed denominator to be expressed with
respect to~$\sigma_y$. There is of course no difference (and hence no
complication) when $\sigma_x=\sigma_y=\mathrm{id}$, as for instance in the
differential case. However, in general it is necessary to impose some further assumption
on the $\partial$-finite ideal~$\mathfrak{a}$ in order for the argument to go
through.

We propose one such assumption in the following definition. It generalizes
the distinction between hypergeometric terms and \emph{proper} hypergeometric
terms known from classical summation theory~\cite{Wilf1992, AbramovPetkovsek2002a}.
At the same time, it refines this notion by distinguishing the free variable $x$
from the summation/integration variable~$y$.

\begin{defi}\label{def:proper}
\begin{enumerate}
\item\label{def:proper:1} A polynomial $u\in K[x,y]$ is called\linebreak
  $y$-\emph{proper} with respect to two endomorphisms $\sigma_x,\sigma_y$
  if $\deg_y\bigl(\leftpart{\fac urx}y\bigr)=\mathrm{O}(1)$ as $r\to\infty$.
\item A $\partial$-finite ideal $\mathfrak{a}\subseteq K(x,y)[\partial_x,\partial_y]=:\set A$ is called \emph{proper}
  (with respect to~$y$) if there exists a $y$-admissible basis $B$ of $\set A/\mathfrak{a}$,
  i.e., the element $1\in\set A/\mathfrak{a}$ is represented by a vector in $K(x)[y]^n$,
  for which the multiplication matrix $\frac1uU$ is such that
  $u$ is $y$-proper with respect to the two endomorphisms $\sigma_x$ and $\sigma_y$ of~$\set A$.
\item\label{def:proper:3}
  Let $B$ and $\frac1uU,\frac1vV$ be as above.
  Let $\eta\in\set N$ be the smallest number such that for all $r\geq1$ there exist
  $f_1,\dots,f_m,g,h\in K[x,y]$,
  $p_1,\dots,p_m,q_1,\dots,q_m\in\set N$, $p_i\geq q_i\geq1$ for $i=1,\dots,m$, with
  \[
    v=\sigma_y(h)\prod_{i=1}^m\fac{f_i}{q_i}y
    \quad\text{\!and}\quad
    \fac urx=g \prod_{i=1}^m\fac{f_i}{p_i}y
  \]
  and $\deg_y(\leftpart gy)\leq\eta$. Then
  \[
    \eta+\max\{\deg_y(v)-1,\penalty-100\deg_y(V)\}
  \]
  is called the \emph{height} of $\mathfrak{a}$ with respect to the basis~$B$.
\item Let $\mathfrak{a}\subseteq\set A$ be a proper $\partial$-finite ideal.
  The \emph{height} of~$\mathfrak{a}$ is defined as the minimum height
  of~$\mathfrak{a}$ with respect to all admissible bases of $\set A/\mathfrak{a}$.
\end{enumerate}
\end{defi}

It is obvious that when $\sigma_x=\sigma_y=\mathrm{id}$, as for instance in the
differential case, then every $\partial$-finite ideal is proper
$\partial$-finite, because in this case $\leftpart{\fac urx}y$ is simply the
squarefree part of~$u$, which does not depend on~$r$. We will further show in
Proposition~\ref{prop:diff} below that in the differential case we always have
$\eta=0$. For the shift case, we will show (Proposition~\ref{prop:shift}) that when $\mathfrak{a}$ is
the annihilator of a hypergeometric term~$h$, then $h$ is proper hypergeometric
if and only if $\mathfrak{a}$ is proper $\partial$-finite with respect to both
$x$ and~$y$.

In part~\ref{def:proper:3} of the definition, observe that the $y$-properness of
$u$ implies that such a number $\eta$ always exist, because a possible (but
perhaps not optimal) choice is $g=\fac urx$, $h=\sigma^{-1}(v)$ and no $f_i$'s
at all (i.e.,~$m=0$). The more complicated condition in the definition allows
for smaller values of $\eta$ by discarding common factors of $u$ and~$v$. This
is desirable because smaller values of $\eta$ will lead to a smaller bound for
the telescoper in Theorem~\ref{thm:bound} below.

\begin{example}
We demonstrate that the right choice of the basis~$B=\{b_1,\dots,b_n\}$ is
crucial in the definition of proper $\partial$-finite ideals. Let $H\in
K(x,y)^{n\times n}$ denote the matrix that realizes the change to a new basis
$\tilde{B}=\{\tilde{b}_1,\dots,\tilde{b}_n\}$, i.e., $\tilde{b}=Hb$ with
$b=(b_1,\dots,b_n)^T$ and $\tilde{b}=(\tilde{b}_1,\dots,\tilde{b}_n)^T$.  If
$M$ is the $\partial_x$-multiplication matrix with respect to~$B$, i.e.,
$\partial_xb=Mb$ as in~\eqref{eq:multmat}, then from
\[
  \partial_x\tilde{b} = \partial_xHb = \big(\sigma_x(H)M+\delta_x(H)\big)b = \tilde{M}\tilde{b}
\]
it follows that $\tilde{M}=\big(\sigma_x(H)M+\delta_x(H)\big)H^{-1}$ is the
$\partial_x$-multi\-plication matrix with respect to~$\tilde{B}$ (and
analogously for $N$ and~$\tilde{N}$). In the following, let $\sigma_x(x)=x+1$,
$\sigma_x(y)=y$, $\sigma_y(x)=x$, $\sigma_y(y)=y+1$, and
$\delta_x=\delta_y=0$.

Consider the classic example of a function that is hypergeometric but not
proper: $f=1/\big(x^2+y^2\big)$; its annihilating ideal~$\mathfrak a$ is
generated by $\big((x+1)^2+y^2\big)\partial_x-x^2-y^2$ and
$\big(x^2+(y+1)^2\big)\partial_y-x^2-y^2$, thus $n=1$.  Choosing $1\in\set
A/\mathfrak a$ as the single basis element~$b_1$, one gets
$M=\big(x^2+y^2\big)/\big((x+1)^2+y^2\big)$; its denominator is clearly not
$y$-proper w.r.t. $\sigma_x$ and $\sigma_y$. Performing the basis change
$H=x^2+y^2$ yields $\tilde{M}=1$ whose denominator is $y$-proper. But still,
the new basis~$\tilde{B}$ does not certify that $f$ is proper
$\partial$-finite since in~$\tilde{B}$ the element $1\in\set A/\mathfrak a$
is represented by $e=H^{-1}$, and therefore $\tilde{B}$ is not admissible.

Next consider the function 
$f=1/(x+y)!+1/(x-y)!$ with the standard monomial
basis $B=\{1,\partial_y\}$, i.e., the basis elements $b_1$ and $b_2$ correspond
to $f(x,y)$ and $f(x,y+1)$. With respect to this basis, the
matrix~$M$ is
\[
  M = \frac1p
  \begin{pmatrix}
   \frac{x^2-2xy+y^2+x-y-1}{y-x+1} & \frac{2y}{y-x-1} \\[1ex]
   \frac{2(y+1)}{x+y+2} & -\frac{x^2+2xy+y^2+3x+3y+1}{x+y+2} \\
  \end{pmatrix}
\]
where $p=y^2-x^2+y-x+1$ is an irreducible quadratic factor.
Again, the basis~$B$ does not certify properness (see also Proposition~\ref{prop:shift}),
but this time we succeed in finding an admissible basis~$\tilde{B}$ which does.
With
\[
  H = \frac1p \begin{pmatrix} (y-x) (x+y+1) & x+y+1 \\ 1 & -(x+y+1) \end{pmatrix}
\]
the multiplication matrices $\tilde{M}=\sigma_x(H)MH^{-1}$ and $\tilde{N}$ are
\[
  \tilde{M} = \begin{pmatrix} \frac{1}{x+y+1} & 0 \\ 0 & \frac{1}{x-y+1} \end{pmatrix}=\frac1uU, \quad
  \tilde{N} = \begin{pmatrix} \frac{1}{x+y+1} & 0 \\ 0 & x-y \end{pmatrix}.
\]
(Note that the diagonal structure of these matrices reveals that the basis elements
$\tilde{b}_1$ and $\tilde{b}_2$ correspond to $1/(x+y)!$ and $1/(x-y)!$,
respectively.) We have now
\[
  \fac urx = \prod_{i=1}^r \big((x+y+i)(x-y+i)\big) = \fac{(x+y+1)(x-y+r)}ry
\]
and therefore $\leftpart{\fac urx}y=(x+y+1)(x-y+r)$ for all $r\geq1$.
\end{example}

\subsection{Main Result}

We are now ready to show the existence of telescopers for proper $\partial$-finite ideals,
and to give an explicit bound on their order.

\begin{theorem}\label{thm:bound}
  Assume that $\mathfrak{a}\subseteq\set A=K(x,y)[\partial_x,\partial_y]$ is proper $\partial$-finite
  with respect to~$y$. Let $\varrho$ be the height of~$\mathfrak{a}$, let
  $n=\dim_{K(x,y)}\set A/\mathfrak{a}$, and let
  $\phi=\dim_{K(x)}\bigl\{W\in\set A/\mathfrak{a} \mid \partial_yW=0\bigr\}$.
  Then there exist $T\in K(x)[\partial_x]\setminus\{0\}$ and $C\in\set A$
  such that $T-\partial_y C\in\mathfrak{a}$ and $\ord(T)\leq n\varrho+\phi$.
\end{theorem}
\begin{proof}
  Let $r=n\varrho+\phi$ and make an ansatz $T=t_0+t_1\partial_x+\cdots+t_r\partial_x^r$
  with undetermined $t_i\in K(x)$ for a telescoper.
  Let $B,\frac1uU,\frac1vV,f_1,\dots,f_m,g,h,q_1,\dots,q_m,p_1,\dots,p_m$ be as in\linebreak
  Definition~\ref{def:proper}. 
  Let $e\in K(x)[y]^n$ be the vector representing $1\in \set A/\mathfrak{a}$ with respect to~$B$,
  let $\gamma=\max\{\deg_y(u),\deg_y(U)\}$, and let $s=\deg_y(e)+r\gamma-\varrho$.
  For the certificate part make an ansatz
  $c=\frac hz\big(\sum_{j=0}^{s}c_{1,j}y^j,\dots,\sum_{j=0}^sc_{n,j}y^j\big)$ with $z$ as in
  Lemma~\ref{lemma:C} and undetermined coefficients $c_{i,j}\in K(x)$.
  A corresponding operator $C\in\set A$ with $Ce=c$ is obtained by $C=c\cdot(B_1,\dots,B_n)^T$
  where $B_i$ is an operator in $\set A$ such that $B_i\cdot1_{\set A/\mathfrak a}=b_i$.

  According to Lemma~\ref{lemma:T}, $Te=\frac1{\fac urx}w_T$ for some vector
  $w_T$ whose entries are linear combinations of the undetermined $t_i$ with
  coefficients in $K(x)[y]$ of degree at most~$\deg_y(e)+r\gamma$.

  According to Lemma~\ref{lemma:C}, $\partial_yc=\frac1{\fac urx}w_C$
  for some vector $w_C$ whose entries are linear combinations of the
  undetermined $c_{i,j}$ with coefficients in $K(x)[y]$ of degree at most
  $s+\varrho=\deg_y(e)+r\gamma$.

  Comparing coefficients with respect to $y$ in all the $n$ coordinates of $w_T$ and $w_C$ gives a linear
  system over $K(x)$ with $n(\deg_y(e)+r\gamma+1)$ equations in $(r+1)+n(s+1)$ unknowns.
  This system has a solution space of dimension at least
  \begin{alignat*}1
    &(r+1)+n(s+1) - n(\deg_y(e)+r\gamma+1)\\
    &=(r+1)+n(\deg_y(e)+r\gamma-\varrho+1) - n(\deg_y(e)+r\gamma+1)\\
    &=r+1-n\varrho=\phi+1.
  \end{alignat*}
  As this is greater than~$\phi$, the solution space must contain at least one vector which corresponds
  to a nonzero operator~$T$.
\end{proof}

Note that the number $\phi$ in Theorem~\ref{thm:bound} is bounded by~$n$.  To
see this, write $W=\sum_{i=1}^n w_i b_i=wb$ for some undetermined $w_i\in K(x,y)$.
Then the requirement $\partial_yW=0$ translates into a first-order linear system of
functional equations $\sigma_y(w)\frac1vV+\delta_y(w)=0$.
It is well known that such a system
can have at most $n$ solution vectors that are linearly independent over the
field of $\sigma_y$-constants, which is $K(x)$ in our case. For a hypergeometric
term $h(x,y)$ we have that $\phi=1$ if and only if $h\in K(x,y)$; this explains
why Theorem~\ref{thm:bound} doesn't exclude such special cases, as opposed to
Apagodu and Zeilberger's theorem, see Section~\ref{sec:az}.

Theorem~\ref{thm:bound} also contains an algorithm for creative telescoping, at
least when $\varrho$ (and the corresponding basis~$B$) and $\phi$ are known or can be computed. In this case, it
suffices to make an ansatz for telescoper and certificate as in the proof,
compare coefficients, and then solve the resulting linear system.

\section{Important Special Cases}

Most important in applications are the differential case (integration) and the
shift case (summation). We will discuss the implications of
Definition~\ref{def:proper} and Theorem~\ref{thm:bound} for these two cases.
Whether a $\partial$-finite ideal is proper or not depends mostly on the
denominators $u$ and $v$ of the multiplication matrices, and not so much on the
numerators $U$ and~$V$. Since $u$ and $v$ are not matrices but only scalar
polynomials, the following discussion is not much different from the
hyperexponential or hypergeometric case.

\subsection{Differential Case}

We consider the case where we act on both variables $x$ and~$y$ with the 
partial derivation, i.e., we have $\sigma_x=\sigma_y=\mathrm{id}$,
$\delta_x=\frac{\d}{\d x}$, and $\delta_y=\frac{\d}{\d y}$.
We have already mentioned that in the differential case every $\partial$-finite
ideal is proper $\partial$-finite. We now show that in this case $u$ and $v$
must be essentially equal. This generalizes Lemma~8 of~\cite{GeddesLeLi2004}. A consequence
is that in part~\ref{def:proper:3} of Definition~\ref{def:proper} we can always take $\eta=0$.

\begin{prop}\label{prop:diff}
  If $\mathfrak{a}\subseteq\set A$ is $\partial$-finite, $B$~is a basis of $\set A/\mathfrak{a}$
  and the multiplication matrices are $\frac1uU,\frac1vV$, then the squarefree part of $u$ in $K(x)[y]$
  divides the squarefree part of $v$ in $K(x)[y]$.
\end{prop}
\begin{proof}
Let~$M := \frac1uU$ and~$N := \frac1vV$. By definition, 
\begin{alignat*}{1}
  \partial_x\partial_y w & = \partial_x(wN + \delta_y(w))  \\
   & = wNM + \delta_x(wN) + \delta_y(w)M + \delta_x\delta_y(w)\\
   & = wNM + w\delta_x(N) + \delta_x(w)N + \delta_y(w)M + \delta_x\delta_y(w).\\
  \partial_y\partial_x w & = \partial_y(wM + \delta_x(w))  \\
   & = wMN + \delta_y(wM) + \delta_x(w)N + \delta_y\delta_x(w)\\
   & = wMN + w\delta_y(M) + \delta_y(w)M + \delta_x(w)N + \delta_y\delta_x(w).
\end{alignat*}
Because of $\partial_x\partial_y=\partial_y\partial_x$, we have the compatibility condition
  \[NM + \delta_x(N) = MN + \delta_y(M).\]
Let $p\in K[x,y]$ with $\deg_y(p)>0$ be an irreducible factor of~$u$, let
  $(i,j)$ be such that $p\nmid U_{i,j}$, and let $m$ be the multiplicity of $p$
  in~$u$. Then the multiplicity of $p$ in the denominator of
  $\delta_y(\frac1uU)_{i,j}$ is $m+1$.
  If $p$ were not also a factor of~$v$, then the multiplicity of $p$ in the
  denominator of $\frac1vV\frac1uU+\delta_x(\frac1vV)-\frac1uU\frac1vV$ could
  be at most~$m$.
\end{proof}

\begin{example}
  Let
  \[
   p=(x^2+x+1)+(2x^2-x+1)y+(x^2-2x+3)y^2
  \]
  and let $\mathfrak{a}\subseteq\set Q(x,y)[\partial_x,\partial_y]$ be the annihilator
  of $f=p^{-1/3}+p^{-1/5}$. Then $n=\dim_{\set Q(x,y)}\set A/\mathfrak{a}=2$
  (since the two summands of~$f$ are hyperexponential but $\set Q(x,y)$-linearly independent)
  and $\phi=\dim_{\set Q(x)}\{W\in\set A/\mathfrak{a}\mid\partial_yW=0\}=0$.  The
  algebra $\set A/\mathfrak{a}$ is isomorphic as $K(x,y)$-vector space to
  $K(x,y)p^{-1/3} + K(x,y)p^{-1/5}$. With respect to the basis
  $B=\{p^{-1/3},p^{-1/5}\}$, the element $1\in\set A/\mathfrak{a}$ is represented
  by the vector $(1,1)\in\set K(x)[y]^2$ and the multiplication matrices are
  \[
    M=\frac{D_x(p)}p\begin{pmatrix} -1/3 &\kern-.5em 0 \\ 0 &\kern-.5em -1/5 \end{pmatrix},\kern.5em
    N=\frac{D_y(p)}p\begin{pmatrix} -1/3 &\kern-.5em 0 \\ 0 &\kern-.5em -1/5 \end{pmatrix}.
  \]
  We can therefore take $u=v=p$ and have
  \[\max\{\deg_y(v)-1,\deg_y(V)\}=1.\]
  Theorem~\ref{thm:bound} predicts a telescoper of order $1\cdot 2+0=2$, and
  it can be confirmed for instance using Chyzak's algorithm that this is in
  fact the minimal order operator.

  Repeating a similar calculation with random polynomials $p$ of $y$-degree $d$
  ($d=2,\dots,5$) and linear combinations $f=p^{e_1}+\dots+p^{e_n}$ with $n$
  rational exponents with pairwise coprime denominators ($n=1,\dots,4$), we found
  the minimal telescopers to be of order $n(d-1)$, in accordance with the bound
  given in Theorem~\ref{thm:bound}.
\end{example}

In the hyperexponential case, Theorem~\ref{thm:bound} reduces to the known bound
quoted at the end of Section~\ref{sec:az}.

\subsection{Shift Case}

In this section, let $\sigma_x$ and $\sigma_y$ denote the standard shifts with
respect to $x$ and~$y$, respectively, i.e., $\sigma_x(x)=x+1$, $\sigma_x(y)=y$,
$\sigma_y(x)=x$, $\sigma_y(y)=y+1$. Let $\delta_y$ be the forward difference
with respect to~$y$ and $\delta_x$ either identically zero or the forward
difference with respect to~$x$.

For a polynomial $p\in K[x,y]$ and $n\in\set N$,
we write $p^{\overline{n}}:=p(p+1)\cdots(p+n-1)$ and
$p^{\underline{n}}:=p(p-1)\cdots(p-n+1)$. Note that these quantities are in
general different from $\fac pnx$ and $\fac pny$.

\begin{prop}\label{prop:shift}
  A $\partial$-finite ideal $\mathfrak{a}$ is proper if and only if there
  exists an admissible basis~$B$ of $\set A/\mathfrak{a}$ for which
  the multiplication matrices
  $\frac1uU$, $\frac1vV$ are such that $u$ is a product of
  integer-linear polynomials.

  More specifically, suppose we can write
  \begin{alignat*}1
    u &= \prod_{i=1}^{k+\ell}(a_ix+b_iy+c_i)^{\overline{a_i}}(a_i'x-b_i'y+c_i')^{\underline{a_i'}} \\
    v &= \sigma_y(h)\prod_{i=1}^{k}(a_ix+b_iy+c_i)^{\overline{b_i}}(a_i'x-b_i'y+c_i')^{\underline{b_i'}}
  \end{alignat*}
  for certain $a_i,b_i,a_i',b_i'\in\set N$ and $c_i,c_i'\in K$, and $h\in K[x,y]$.
  If $\eta$ is as in part~\ref{def:proper:3} of Definition~\ref{def:proper},
  then $\eta\leq\sum_{i=k+1}^{k+\ell}(b_i+b_i')$.
\end{prop}
\begin{proof}
  Suppose that $\mathfrak{a}$ is proper $\partial$-finite.
  Let $p$ be an irreducible factor of~$u$
  such that both~$\deg_x(p)$ and~$\deg_y(p)$ are nonzero.
  We have $\fac urx \mid \fac u{r+1}x$ for all $r\geq0$.
  By the condition in part~\ref{def:proper:1} of Def.~\ref{def:proper},
  the set of irreducible factors of the left borders of $\fac urx$
  for all $r\geq0$ is finite.
  Therefore, there is at least one positive integer~$s$
  such that $\sigma_x^{s}(p) = \sigma_y^{t}(p)$ for some $t\in \set Z$.
  By Corollary~1 in~\cite[page 400]{AbramovPetkovsek2002a},
  $p$ is integer-linear.

  Conversely, if $u$ is a product of integer-linear polynomials $a_ix+b_iy+c_i$,
  then it is sufficient to prove the more specific claim, because if $u$ and $v$
  are not given in this form, we can multiply both $u$ and $U$ with the missing
  factors such as to complete the rising and falling factorials.

  In order to keep the notation simple, let us only discuss the factors
  $(a_ix+b_iy+c_i)$. An analogous argument applies to the other factors
  $(a_i'x-b_i'y+c_i')$.

  Let $r\geq1$. For fixed $i$, write $p_i=a_ix+b_iy+c_i$ and let $s,t\in\set N$ be such
  that $a_ir=sb_i+t$. Then
  \begin{alignat*}1
    &\fac{(a_ix+b_iy+c_i)^{\overline{a_i}}}rx
    =p_i^{\overline{a_ir}}\\
    ={}&(p_i)(p_i+1)\cdots\cdots(p_i+b_i-1)\\
    &\times(p_i+b_i)(p_i+b_i+1)\cdots(p_i+2b_i-1)\\
    &\vdots\\
    &\times(p_i+(s-1)b_i)(p_i+(s-1)b_i+1)\cdots(p_i+sb_i-1)\\
    &\times(p_i+sb_i)\cdots(p_i+sb_i+t)\\
    ={}&\prod_{j=0}^t \fac{p_i+j}{s+1}y
      \prod_{j=t+1}^{b_i-1} \fac{p_i+j}sy.
  \end{alignat*}
  Therefore, if we choose $f_1,\dots,f_m$ to be all the linear factors $p_i+j$ ($i=1,\dots,k$;
   $j=0,\dots,b_i-1$) and set
  \[
    g=\prod_{i=k+1}^{k+\ell} p_i^{\overline{a_i r}},
  \]
  then we will have $u=g\prod_{i=1}^m\fac{f_i}{s_i}y$ for certain $s_i\in\set N$
  with $s_i\geq\lfloor a_ir/b_i\rfloor$
  and
  \[
    \leftpart gy = \prod_{i=k+1}^{k+\ell} p_i^{\overline{b_i}},
  \]
  the $y$-degree of which is $\sum_{i=k+1}^{k+\ell}b_i$, as claimed.
\end{proof}

\begin{example}
  For fixed $n\geq0$ and~$\varrho$, the annihilator $\mathfrak{a}$ of the function
  \[
    f(x,y)=\frac{1+2^y+3^y+\cdots+n^y}{\Gamma(x+\varrho y)}
  \]
  is proper $\partial$-finite with $\eta=0$, dimension~$n$, and height~$\varrho$.
  As the exponential terms $k^y$ ($k=1,\dots,n$) are algebraically independent over $K(x,y)$,
  there is no nontrivial $W\in\set A/\mathfrak{a}$ for which $\partial_yW=0$.
  Therefore $\phi=0$. The minimal telescoper for $f(x,y)$ is
  \[
    T = (\partial_x^\varrho-1)(\partial_x^\varrho-2)\cdots(\partial_x^\varrho-n)
  \]
  and its order $n\varrho=n\varrho+\phi$ matches the bound of Theorem~\ref{thm:bound}.
  The corresponding certificate~$C$ cannot be written in such a nice form and
  is therefore not displayed here.
\end{example}

For the hypergeometric case, our bound does not exactly reduce to the known
bounds stated in Section~\ref{sec:az} for this case. Our bound is at the same
time better and worse than the old bound. It is worse because for the
hypergeometric case it turns out that because of an additional cancellation the
term $\eta=\deg_y(\leftpart gy)$ does not contribute to the order. It is
slightly better because we work in the Ore algebra where $\partial_y$ represents
the forward difference rather than the shift operator, and for certain
hypergeometric terms, it turns out that this improves the bound by~$1$. For
example, for the hypergeometric term $(x+3y+1)!/(x+3y+\sqrt2)!$ our bound
evaluates to~$2$, which is indeed the order of the minimal telescoper, while the
bound of Section~\ref{sec:az} only predicts a telescoper of order~$3$.

\subsection{Mixed and Other Cases}

Thanks to the generality in which we stated our results in
Section~\ref{sec:general} we can not only deal with the pure differential or
pure shift cases discussed above, but also with mixed cases where the two
indeterminates $x$ and $y$ are different in nature (discrete versus
continuous). In these cases, a necessary condition for an ideal to be proper
$\partial$-finite is that the polynomial $u$ is split, i.e., that it can be
written as $u(x,y)=u_1(x)u_2(y)$. A polynomial that violates this condition
can never be $y$-proper. We now give an example where $x$ is a continuous
variable and $y$ is discrete, corresponding to a definite sum over~$y$ for
which a differential equation in~$x$ is sought.

\begin{example}
Let $\set A=K(x,y)[\partial_x,\partial_y]$ be the Ore algebra given by
$\sigma_x=\mathrm{id}$, $\delta_x=\frac{\d}{\d x}$, $\sigma_y(y)=y+1$,
and $\delta_y=\sigma_y-\mathrm{id}$. With respect to this algebra
each member of the family
\[
  f_k(x,y) = (y+1)^{-k} J_y(x),\quad k\in\set N,
\]
involving the Bessel function of the first kind, is $\partial$-finite.
For any fixed~$k$, the annihilator~$\mathfrak{a}$ of $f_k(x,y)$ is generated
by two operators, one of which corresponds to the famous Bessel differential
equation $x^2\partial_x^2+x\partial_x+x^2-y^2$, and we have
$n=\dim_{K(x,y)}(\set A/\mathfrak{a})=2$.
As a basis for $\set A/\mathfrak{a}$ we choose the two monomials $1$ and $\partial_x$
so that the multiplication matrices are
\begin{alignat*}1
  U &= \begin{pmatrix} 0 & x^2 \\ y^2-x^2 & -x \end{pmatrix}\\
  V &= \begin{pmatrix} xy(y+1)^k - x^2(y+2)^k & -x^2(y+1)^k \\ (y+1)^k (x^2-y^2-y) & x(y+1)^{k+1} - x^2(y+2)^k \end{pmatrix}
\end{alignat*}
with denominators $u=x^2$ and $v=x^2 (y+2)^k$. Obviously $u$ is $y$-proper
and therefore the height of $\mathfrak{a}$ is (at most) $\max\{\deg_y(v)-1,\deg_y(V)\}=k+2$.
Taking $\phi=0$ into account, Theorem~\ref{thm:bound} produces the bound $2(k+2)$
for the order of the telescoper. In contrast, the minimal telescoper conjecturally is of order
$2k+1$ (we verified this for $0\leq k\leq 20$), so our bound overshoots by~$3$.
\end{example}

Last but not least let us emphasize that all our results also apply to the $q$-case,
where $\sigma_y(y)=qy$ and $\delta_y=\sigma_y-\mathrm{id}$; it is very much
analogous to the shift case.

\section{Conclusion and Open Questions}

We have shown that the reasoning of Apagodu and Zeilberger applies in the
general setting of $\partial$-finite ideals in Ore algebras.

As a sufficient condition for guaranteeing the existence of a telescoper, we
have introduced the notion of ``proper'' $\partial$-finite ideals, in analogy
with the notion of proper hypergeometric terms in classical summation
theory. For hypergeometric terms, Wilf and Zeilberger conjectured in 1992 that
they are proper if and only if they are holonomic. Slightly modified versions
of this conjecture were proved
independently~\cite{Payne97,AbramovPetkovsek2002a} for the shift case, and
recently~\cite{CKP2013} for general hypergeometric terms.  It is now tempting
to conjecture that, more generally, a $\partial$-finite ideal is proper if and
only if it is holonomic.

For the hypergeometric case, Abramov~\cite{Abramov2003} pointed out that
proper is only a sufficient condition, but it is not necessary for the
existence of a telescoper, and he formulates a finer condition which is
necessary and sufficient. Abramov's existence criterion has been extended to
the $q$-shift case and mixed cases~\cite{Chen2005, ChenThesis}. It would be
interesting to have an analogous result for the $\partial$-finite case.

\subsection*{Acknowledgments}
We would like to thank the anonymous referees for their diligent work
and valuable comments.

\bibliographystyle{plain}

\bibliography{GAZ}
\end{document}